\documentclass{ieeeaccess}
\usepackage{cite}
\usepackage{amsmath,amssymb,amsfonts}
\usepackage{algorithmic}
\usepackage{graphicx}
\usepackage{textcomp}
\usepackage{bm}
\usepackage{amsthm}

\newtheorem{definition}{Definition}
\newtheorem{proposition}{Proposition}
\newtheorem{lemma}{Lemma}
\newtheorem{theorem}{Theorem}

\def\BibTeX{{\rm B\kern-.05em{\sc i\kern-.025em b}\kern-.08em
    T\kern-.1667em\lower.7ex\hbox{E}\kern-.125emX}}
\begin{document}
\history{Date of publication xxxx 00, 0000, date of current version xxxx 00, 0000.}
\doi{10.1109/ACCESS.2023.3235922}

\title{Unexploitable games and unbeatable strategies}
\author{\uppercase{Masahiko Ueda}\authorrefmark{1}}

\address[1]{Graduate School of Sciences and Technology for Innovation, Yamaguchi University, Yamaguchi 753-8511, Japan (e-mail: m.ueda@yamaguchi-u.ac.jp)}
\tfootnote{This study was supported by JSPS KAKENHI Grant Number JP20K19884 and Inamori Research Grants.}

\markboth
{Author \headeretal: Preparation of Papers for IEEE TRANSACTIONS and JOURNALS}
{Author \headeretal: Preparation of Papers for IEEE TRANSACTIONS and JOURNALS}

\corresp{Corresponding author: Masahiko Ueda (e-mail: m.ueda@yamaguchi-u.ac.jp).}

\begin{abstract}
Imitation is simple behavior which uses successful actions of others in order to deal with one's own problems.
Because success of imitation generally depends on whether profit of an imitating agent coincides with those of other agents or not, game theory is suitable for specifying situations where imitation can be successful.
One of the concepts describing successfulness of imitation in repeated two-player symmetric games is unbeatability.
For infinitely repeated two-player symmetric games, a necessary and sufficient condition for some imitation strategy to be unbeatable was specified.
However, situations where imitation can be unbeatable in multi-player games are still not clear.
In order to analyze successfulness of imitation in multi-player situations, here we introduce a class of totally symmetric games called unexploitable games, which is a natural extension of two-player symmetric games without exploitation cycles.
We then prove that, for infinitely repeated unexploitable games, there exist unbeatable imitation strategies.
Furthermore, we also prove that, for infinitely repeated non-trivial unexploitable games, there exist unbeatable zero-determinant strategies, which unilaterally enforce some relationships on payoffs of players.
These claims are demonstrated in the public goods game, which is the simplest unexploitable game.
These results show that there are situations where imitation is unbeatable even in multi-player games.
\end{abstract}


\titlepgskip=-15pt

\maketitle

\section{Introduction}
\label{sec:intro}
\PARstart{I}{mitation} is simple behavior which uses successful actions of others in order to deal with one's own problems.
Copying the behavior of others is sometimes successful in human society \cite{RBCet2010}.
In biological systems, a mechanism such that genes are passed from parents to offspring is naturally adopted \cite{HarCla1997}.
In economics, it has been discussed that equilibria are realized not due to rational thinking but due to imitation \cite{Veg1997,Sch1998}.
On the other hand, piracy is generally prohibited in creative activities \cite{BelPei2010}.
When there are multiple agents who imitate others, complicated dynamics can occur \cite{SuzKan1994}.
In general, success of imitation depends on whether profit of a copying agent coincides with those of other agents or not, which is a subject of game theory \cite{FudTir1991}.

When we restrict our attention to repeated two-player symmetric games, one of the candidates which characterize successfulness of imitation is unbeatability \cite{DOS2012b,DOS2014}.
Unbeatability literally means that the imitating agent cannot be beaten, or, cannot be exploited.
For infinitely repeated two-player symmetric games, it has been known that there exists an unbeatable imitation strategy, called the Imitate-If-Better (IIB) strategy, if and only if the game does not contain any cycles similar to the rock-paper-scissors cycle \cite{DOS2012b}.
Furthermore, it has also been proved that the Tif-for-Tat (TFT) strategy \cite{RCO1965,AxeHam1981}, which imitates the opponent's previous action, is unbeatable if and only if the game is a potential game \cite{DOS2014}.

Recently, it has been pointed out that the existence of unbeatable imitation strategies may be related to the existence of unbeatable zero-determinant (ZD) strategies \cite{Ued2022b}.
ZD strategies are a class of memory-one strategies in infinitely repeated games, which unilaterally enforce linear relationships between payoffs of players \cite{PreDys2012}.
Although ZD strategies were originally introduced in the prisoner's dilemma game, they have been extended to broader situations \cite{HRZ2015,HTS2015,HDND2016,McAHau2017,MamIch2020}.
Furthermore, the control ability of ZD strategies has also been extended \cite{Ued2021,Ued2021b,Ued2022c}.
As application studies, ZD strategies have been applied to several problems in the field of information and communications technology, such as resource sharing in wireless networks \cite{DKL2014,ZNSet2016}, mining in blockchain \cite{TLYet2019}, crowdsourcing \cite{TLCet2019}, the Internet of Things \cite{WSHet2019}, cloud computing \cite{ZTYet2020}, and data trading \cite{WSHet2021}.
In terms of unbeatable strategies, the author proved that the unbeatable TFT is a ZD strategy in two-player symmetric games \cite{Ued2022}.
In addition, for two-player symmetric games where the unbeatable imitation strategy exists, unbeatable ZD strategies also exist \cite{Ued2022b}.
These results suggest that there may be some relation between the existence of unbeatable imitation strategies and the existence of unbeatable ZD strategies, even in multi-player symmetric games.

The purpose of this paper is investigating the existence of unbeatable imitation strategies and unbeatable ZD strategies in multi-player totally symmetric games \cite{Pla2023}.
Concretely, we introduce a class of multi-player totally symmetric games called unexploitable games.
Unexploitable games are an extension of two-player symmetric games without generalized-rock-paper-scissors cycles to multi-player case.
We show that the unexploitable property of the stage game is a sufficient condition for the existence of unbeatable ZD strategies and the unbeatable IIB strategy.
We also explain these results in the public goods game \cite{RotKag1995}, which is the simplest example of unexploitable games.

This paper is organized as follows.
In Section \ref{sec:setup}, we define a model of infinitely repeated multi-player totally symmetric games.
In Section \ref{sec:preliminaries}, we introduce basic concepts used in the later sections, such as ZD strategies, unbeatable strategies, and imitation strategies.
In Section \ref{sec:results}, we introduce the concept of unexploitable games, and provide our main results on the existence of unbeatable ZD strategies and unbeatable imitation strategies in unexploitable games.
In Section \ref{sec:example}, we demonstrate our results in the public goods game.
Section \ref{sec:conclusion} is devoted to concluding remarks.

\section{Setup}
\label{sec:setup}
We consider an infinitely repeated game with $N$ players \cite{MaiSam2006}.
The stage game is $G:= \left( \mathcal{N}, \left\{ A_j \right\}_{j\in \mathcal{N}}, \left\{ s_j \right\}_{j\in \mathcal{N}} \right)$, where $\mathcal{N}:=\left\{ 1, \cdots, N \right\}$ is the set of players, $A_j$ is the set of actions of player $j\in \mathcal{N}$, and $s_j: \prod_{k=1}^N A_k \rightarrow \mathbb{R}$ is the payoff of player $j\in \mathcal{N}$.
We collectively write $\mathcal{A}:=\prod_{j=1}^N A_j$ and $\bm{a}:=\left( a_1, \cdots,  a_N \right)\in \mathcal{A}$, and call $\bm{a}$ an action profile.
We write a probability $L$-simplex by $\Delta_L$.
We also introduce the notations $a_{-j} := \left( a_1, \cdots, a_{j-1}, a_{j+1}, \cdots, a_N \right)\in \prod_{k\neq j} A_k$ and $-j:=\mathcal{N}\backslash \left\{ j \right\}$.
(Below, when we want to emphasize the action of player $j$ in $\bm{a}$, we write $\bm{a}=\left( a_j, a_{-j} \right)$.)
We assume that the stage game is \emph{totally symmetric} \cite{Pla2023}, that is, $A_j=A$ $(\forall j\in \mathcal{N})$ and for every permutation $\pi$ on $\mathcal{N}$,
\begin{align}
 s_{\pi(j)} \left( \bm{a} \right) &= s_j \left( \bm{a}_{\pi} \right)
\end{align}
for any $j\in \mathcal{N}$ and for any $\bm{a}\in \mathcal{A}$, where $\bm{a}_{\pi}:= \left( a_{\pi(1)}, \cdots, a_{\pi(N)} \right)$ and $A$ is some set.
We also assume that $A$ is finite, and write $A=\left\{ 1, \cdots, L \right\}$, where $L\in \mathbb{N}$ is the number of actions.

We repeat the stage game $G$ infinitely.
We write an action of player $j$ at round $t\geq 1$ as $a_j^{(t)}$.
The behavior strategy of player $j\in \mathcal{N}$ is described as $\mathcal{T}_j := \left\{ T^{(t)}_j \right\}_{t=1}^\infty$, where $T^{(t)}_j: \mathcal{A}^{t-1} \to \Delta_{L}$ is the conditional probability at $t$-th round.
We write the expectation of the quantity $B$ with respect to strategies of all players by $\mathbb{E}[B]$.
We consider the case that the payoff of player $j\in \mathcal{N}$ in the infinitely repeated game is given by
\begin{align}
 \mathcal{S}_j &:= \lim_{T\rightarrow \infty} \frac{1}{T} \sum_{t=1}^T \mathbb{E} \left[ s_j\left( \bm{a}^{(t)} \right) \right],
 \label{eq:payoff_repeated}
\end{align}
that is, we consider a repeated game with no discounting.

We remark that we frequently use the prime symbol to generate more variables whose types are the same as ones without the prime symbol.

\section{Preliminaries}
\label{sec:preliminaries}
In this section, we introduce several concepts used in later sections.
We define the marginal distribution at $t$-th round obtained from the joint distribution of action profiles
\begin{align}
 P_t \left( \bm{a}^{(t)} \right) &:= \sum_{\bm{a}^{(t-1)}} \cdots \sum_{\bm{a}^{(1)}} P\left( \bm{a}^{(t)}, \cdots, \bm{a}^{(1)} \right),
\end{align}
for all $\bm{a}^{(t)} \in \mathcal{A}$, and the limit distribution
\begin{align}
 P^* \left( \bm{a} \right) &:= \lim_{T\rightarrow \infty} \frac{1}{T} \sum_{t=1}^T P_{t} \left( \bm{a} \right) \quad (\forall \bm{a}).
\end{align}
We also write the expectation with respect to the limit distribution $P^*$ by $\left\langle \cdots \right\rangle^{*}$.
It should be noted that $\mathcal{S}_k=\left\langle s_k \right\rangle^{*}$ $(\forall k \in \mathcal{N})$.

\subsection{Zero-determinant strategies}
\label{subsec:ZD}
A \emph{time-independent memory-one strategy} of player $j\in \mathcal{N}$ is defined as a strategy such that
\begin{align}
 T^{(t)}_j\left( a_j^{(t)} | \bm{a}^{(t-1)}, \cdots, \bm{a}^{(1)} \right) &= T_j \left( a_j^{(t)} | \bm{a}^{(t-1)} \right)
\end{align}
for $\forall t\geq 2$ and for all $a_j^{(t)}$, $\bm{a}^{(t-1)}$, $\cdots$, $\bm{a}^{(1)}$, where $T_j: \mathcal{A} \to \Delta_{L}$.
For time-independent memory-one strategies $T_j$ of player $j$, we introduce
\begin{align}
 \hat{T}_j\left( a_j | \bm{a}^{\prime} \right) &:= T_j\left( a_j | \bm{a}^{\prime} \right) -  \delta_{a_j, a^{\prime}_j} \quad \left( \forall a_j \in A, \forall \bm{a}^{\prime} \in \mathcal{A} \right),
 \label{eq:PD}
\end{align}
where $\delta_{a, a^\prime}$ is the Kronecker delta.
These quantities describe the difference between the strategy $T_j$ and the Repeat strategy $\delta_{a_j, a^{\prime}_j}$, and called as the \emph{Press-Dyson vectors} \cite{Aki2016,McAHau2016,UedTan2020}.
It should be noted that, due to the properties of the conditional probability $T_j$, the Press-Dyson vectors satisfy several relations.
First, they satisfy
\begin{align}
 \sum_{a_j} \hat{T}_j \left( a_j | \bm{a}^\prime \right) &= 0 \quad \left( \forall \bm{a}^\prime \right)
 \label{eq:PD_normalized}
\end{align}
due to the normalization condition of $T_j$.
Second, they satisfy
\begin{align}
  & \left\{
  \begin{array}{ll}
    -1 \leq \hat{T}_j \left( a_j | \bm{a}^\prime \right) \leq 0 & \left(a_j = a^\prime_j \right) \\
    0 \leq \hat{T}_j \left( a_j | \bm{a}^\prime \right) \leq 1 & \left(a_j \neq a^\prime_j \right)
  \end{array}
  \right.
  \label{eq:PD_range}
\end{align}
for all $a_j$ and all $\bm{a}^\prime$, because $T_j$ takes values in $[0,1]$.

It has been known that the Press-Dyson vectors satisfy the relation called Akin's lemma.
\begin{lemma}[\cite{Aki2016,Ued2022}]
\label{lem:Akin}
The Press-Dyson vectors of player $j$ using a time-independent memory-one strategy satisfy
\begin{align}
 \sum_{\bm{a}^\prime} P^* \left( \bm{a}^\prime \right) \hat{T}_j\left( a_j | \bm{a}^{\prime} \right) &= 0
 \label{eq:Akin}
\end{align}
for all $a_j$.
\end{lemma}

Now we introduce an extended version \cite{Ued2021,Ued2022c} of zero-determinant strategies \cite{PreDys2012}.
\begin{definition}
\label{def:ZDS}
A time-independent memory-one strategy of player $j$ is an (extended) \emph{zero-determinant (ZD) strategy} controlling the quantity $B:\mathcal{A}\rightarrow \mathbb{R}$ when its Press-Dyson vectors can be written in the form
\begin{align}
 \sum_{a_j} c_{a_j} \hat{T}_j\left( a_j | \bm{a}^{\prime} \right) &= B \left( \bm{a}^{\prime} \right) \quad \left( \forall \bm{a}^{\prime} \right)
 \label{eq:ZDS}
\end{align}
with some nontrivial coefficients $\left\{ c_{a_j} \right\}$ (that is, not $c_1=\cdots=c_{L}=\mathrm{const.}$) and $B$ is not identically zero.
\end{definition}
In other words, in ZD strategies controlling $B$, $B$ is described by a linear combination of the Press-Dyson vectors.
Because of Lemma \ref{lem:Akin}, ZD strategies unilaterally enforce
\begin{align}
 \left\langle B \right\rangle^{*} &= 0.
 \label{eq:ZDS_linear}
\end{align}
A necessary and sufficient condition for the existence of ZD strategies is given by the following proposition.
\begin{proposition}[\cite{Ued2022b}]
\label{prop:existence}
A ZD strategy of player $j$ controlling $B$ exists if and only if there exist two different actions $\overline{a}, \underline{a} \in A$ of player $j$ such that
\begin{align}
 B \left( \overline{a}, a_{-j} \right) &\geq 0 \quad \left( \forall a_{-j} \right) \nonumber \\
 B \left( \underline{a}, a_{-j} \right) &\leq 0 \quad \left( \forall a_{-j} \right),
 \label{eq:condition_exsitence}
\end{align}
and $B$ is not identically zero.
\end{proposition}

\subsection{Unbeatable strategies}
\label{subsec:unbeatable}
Unbeatable strategies are literally those which cannot be beaten by other players in repeated games.
Unbeatable strategies were originally introduced in two-player symmetric games \cite{DOS2012b,DOS2014}.
(They are also called as rival strategies \cite{HTS2015}.)
A multi-player version of unbeatable strategies has also been investigated in the repeated public goods game, which also achieves mutual cooperation \cite{MurBae2018,MurBae2021}.
Here we define unbeatable strategies in multi-player totally symmetric games, following the previous studies.
\begin{definition}
\label{def:unbeatable}
A strategy $\mathcal{T}_j$ of player $j\in \mathcal{N}$ in repeated totally symmetric games is \emph{unbeatable} if
\begin{align}
 \mathcal{S}_j &\geq \mathcal{S}_k \quad \left( \forall k \neq j \right)
\end{align}
for all strategies $\mathcal{T}_k$ $(\forall k\neq j)$.
\end{definition}
Several examples of unbeatable strategies have been found in the prisoner's dilemma game \cite{PreDys2012,YBC2017,MurBae2020}, in two-player symmetric potential games \cite{DOS2014,Ued2022}, in two-player symmetric games with no generalized rock-paper-scissors cycles \cite{DOS2012b,Ued2022b}, and in the public goods game \cite{MurBae2018,MurBae2021}.
In addition, a uniform strategy in the repeated rock-paper-scissors game is also an unbeatable strategy.

\subsection{Imitation strategies}
\label{subsec:imitation}
Imitation strategies are a class of strategies in repeated games where an action is chosen from the set of actions used in the previous round.
A typical example is the Tit-for-Tat (TFT) strategy \cite{RCO1965,AxeHam1981,DOS2014} in two-player symmetric games, which returns the previous action of the opponent.
Another example is the Imitate-If-Better (IIB) strategy, which imitates the opponent's previous action if the player was beaten in the previous round \cite{DOS2012b}.
Extension of IIB to multi-player totally-symmetric games is straightforward.
\begin{definition}
\label{def:IIB}
The \emph{Imitate-If-Better} (IIB) strategy of player $j\in \mathcal{N}$ is a time-independent memory-one strategy such that
\begin{align}
 T_j\left( a_j | \bm{a}^{\prime} \right) &= \delta_{a_j, a^\prime_{\arg\max_{k\neq j}s_k(\bm{a}^\prime)}} \mathbb{I} \left( s_j(\bm{a}^\prime) < \max_{k\neq j}s_k(\bm{a}^\prime) \right) \nonumber \\
 & \quad + \delta_{a_j, a^\prime_j} \mathbb{I} \left( s_j(\bm{a}^\prime) \geq \max_{k\neq j}s_k(\bm{a}^\prime) \right) \nonumber \\
 & \quad \left( \forall a_j, \forall \bm{a}^\prime \right),
 \label{eq:IIB}
\end{align}
where $\mathbb{I}(\cdots)$ is an indicator function which returns 1 if $\cdots$ holds and 0 otherwise.
It should be noted that, if several players are contained in $\arg\max_{k\neq j}s_k(\bm{a}^\prime)$, each player in the set is chosen with equal probability.
\end{definition}
For $N=2$, it has been known that IIB is unbeatable if and only if the state game does not contain any generalized rock-paper-scissors cycles \cite{DOS2012b}.

\section{Results}
\label{sec:results}
The purpose of this paper is to find situations where unbeatable ZD strategies or unbeatable imitation strategies exist for the case $N>2$.
For this purpose, as an extension of games without generalized rock-paper-scissors cycles in the case $N=2$ \cite{DOS2012a,DOS2012b}, we introduce the concept of unexploitable games.
First, we define
\begin{align}
 G_j\left( \bm{a} \right) &:= \max_{k\neq j}s_k\left( \bm{a} \right) \quad (\forall \bm{a})
 \label{eq:G}
\end{align}
for all $j\in \mathcal{N}$.
\begin{definition}
\label{def:unexploitable}
A stage game is an \emph{unexploitable game} if, for all subset of the action space $\forall A^\prime \subseteq A$, there exists at least an action $a^*\left( A^\prime \right) \in A^\prime$ of player $j$ such that
\begin{align}
 & s_j \left( a^*\left( A^\prime \right), a_{-j} \right) - G_j\left( a^*\left( A^\prime \right), a_{-j} \right) \geq 0 \nonumber \\
 & \quad \left( \forall a_{-j}\in A^{\prime N-1} \right).
 \label{eq:unex}
\end{align}
\end{definition}
It should be noted that the definition does not depend on $j$ because the stage game is totally symmetric.
We also call an unexploitable game \emph{trivial} if $s_k=s$ $(\forall k \in \mathcal{N})$ with a function $s$.
For trivial unexploitable games, since the payoffs of all players are always the same, all strategies are trivially unbeatable.
Below we mainly focus on non-trivial unexploitable games.

\subsection{Unbeatable ZD strategies in unexploitable games}
\label{subsec:unbeatable_ZD}
We first prove that unbeatable ZD strategies exist in non-trivial unexploitable games.
\begin{theorem}
\label{thm:unbeatable_ZD}
If the stage game is a non-trivial unexploitable game, then an unbeatable ZD strategy exists.
\end{theorem}

\begin{proof}
Because the game is unexploitable, we can recursively define
\begin{align}
 A^{(l)} &:= A\backslash \left\{ a^*\left( A^{(1)} \right), \cdots, a^*\left( A^{(l-1)} \right) \right\} \quad (1\leq l\leq L).
 \label{eq:Al}
\end{align}
(When there are several candidates for each $a^*\left( A^{(l)} \right)$, we choose one of them.)
Particularly, $a^*\left( A^{(1)} \right)=a^*\left( A \right)$ is the strongest action
\begin{align}
 & s_j \left( a^*\left( A^{(1)} \right), a_{-j} \right) - G_j\left( a^*\left( A^{(1)} \right), a_{-j} \right) \geq 0 \nonumber \\
 & \quad \left( \forall a_{-j}\in A^{N-1} \right).
 \label{eq:strongest}
\end{align}
By the definition, $\left\{ a^*\left( A^{(L)} \right) \right\}=A^{(L)}\subset A^{(L-1)}\subset \cdots \subset A^{(1)}=A$.

We now prove that $a^*\left( A^{(L)} \right)$ is the weakest action:
\begin{align}
 & s_j \left( a^*\left( A^{(L)} \right), a_{-j} \right) - G_j\left( a^*\left( A^{(L)} \right), a_{-j} \right) \leq 0 \nonumber \\
 & \quad \left( \forall a_{-j}\in A^{N-1} \right).
 \label{eq:weakest}
\end{align}
Assume to the contrary that
\begin{align}
 & s_j \left( a^*\left( A^{(L)} \right), a_{-j} \right) - G_j\left( a^*\left( A^{(L)} \right), a_{-j} \right) > 0 \nonumber \\
 & \quad \left( \exists a_{-j}\in A^{N-1} \right).
\end{align}
We write this $a_{-j}$ as $\tilde{a}_{-j}$ and define $\tilde{\bm{a}}:=\left( a^*\left( A^{(L)} \right), \tilde{a}_{-j} \right)$.
If all players use $a^*\left( A^{(L)} \right)$, we obtain
\begin{align}
 & s_j \left( a^*\left( A^{(L)} \right), \cdots, a^*\left( A^{(L)} \right) \right) \nonumber \\
 & \quad - G_j\left( a^*\left( A^{(L)} \right), \cdots, a^*\left( A^{(L)} \right) \right) = 0
\end{align}
because the game is totally symmetric.
Therefore, at least one player uses the action which is not equal to $a^*\left( A^{(L)} \right)$ in $\tilde{\bm{a}}$.
We consider the minimal $A^{(l)}$ such that $\tilde{\bm{a}} \in A^{(l)N}$.
(It should be noted that $1\leq l \leq L-1$.)
Then, there exists at least one player $j^\prime \neq j$ such that $\tilde{a}_{j^\prime} = a^*\left( A^{(l)} \right)$.
Due to the symmetry of the game and the definition of $a^*\left( A^{(l)} \right)$,
\begin{align}
 s_{j^\prime} \left( \tilde{\bm{a}} \right) - G_{j^\prime} \left( \tilde{\bm{a}} \right) &\geq 0.
\end{align}
Then, we obtain
\begin{align}
 s_{j} \left( \tilde{\bm{a}} \right) &> G_{j} \left( \tilde{\bm{a}} \right) = \max_{k\neq j} s_{k} \left( \tilde{\bm{a}} \right) \nonumber \\
 &\geq s_{j^\prime} \left( \tilde{\bm{a}} \right) \nonumber \\
 &\geq G_{j^\prime} \left( \tilde{\bm{a}} \right) = \max_{k\neq j^\prime} s_{k} \left( \tilde{\bm{a}} \right) \nonumber \\
 &\geq s_{j} \left( \tilde{\bm{a}} \right),
\end{align}
leading to contradiction.
Therefore, we obtain Eq. (\ref{eq:weakest}).

Now, we define
\begin{align}
 B\left( \bm{a} \right) &:= s_j \left( \bm{a} \right) - G_j \left( \bm{a} \right) \quad (\forall \bm{a}\in \mathcal{A}).
\end{align}
Because we consider a non-trivial unexploitable game, $B$ is not identically zero.
Then, by writing $\overline{a}=a^*\left( A^{(1)} \right)$ and $\underline{a}=a^*\left( A^{(L)} \right)$, we can apply Proposition \ref{prop:existence}, and find the existence of a ZD strategy of player $j$ controlling $B$.
Such ZD strategy unilaterally enforces $\left\langle s_j \right\rangle^{*}=\left\langle G_j \right\rangle^{*}$.

Finally, since
\begin{align}
 \max_{k\neq j} \mathcal{S}_k &= \max_{k\neq j} \left\langle s_k \right\rangle^{*} \nonumber \\
 &= \max_{k\neq j} \sum_{\bm{a}} P^* \left( \bm{a} \right) s_k \left( \bm{a} \right) \nonumber \\
 &= \sum_{\bm{a}} P^* \left( \bm{a} \right) s_{\arg \max_{k\neq j} \sum_{\bm{a}^\prime} P^* \left( \bm{a}^\prime \right) s_k \left( \bm{a}^\prime \right)} \left( \bm{a} \right) \nonumber \\
 &\leq \sum_{\bm{a}} P^* \left( \bm{a} \right) \max_{k\neq j} s_k \left( \bm{a} \right) \nonumber \\
 &= \left\langle G_j \right\rangle^{*},
 \label{eq:max_G}
\end{align}
we conclude that this ZD strategy unilaterally enforces $\mathcal{S}_j \geq \max_{k\neq j} \mathcal{S}_k$, that is, it is unbeatable.
\end{proof}

We remark that the converse of Theorem \ref{thm:unbeatable_ZD} does not hold, even in the case of $N=2$ \cite{Ued2022b}.

\subsection{Unbeatable imitation in unexploitable games}
\label{subsec:unbeatable_IIB}
Next, we prove that IIB (Definition \ref{def:IIB}) is unbeatable in unexploitable games.
\begin{theorem}
\label{thm:unbeatable_IIB}
If the stage game is an unexploitable game, then IIB is unbeatable.
\end{theorem}

\begin{proof}
We define $A^{(l)}$ $(1\leq l \leq L)$ as in Eq. (\ref{eq:Al}), and call $a^*\left( A^{(l)} \right)$ with smaller $l$ ``stronger''.
We consider the situation that player $j$ uses IIB (\ref{eq:IIB}).
Let $A^\infty \subseteq A$ be a set of actions which are played by players $-j$ an infinite number of times.
We now show that the action of player $j$ must converge to an action which is equivalent to or stronger than $a^*\left( A^\infty \right)$.
We consider the minimal $A^{(l)}$ such that $A^\infty \subseteq A^{(l)}$.
Trivially, $a^*\left( A^\infty \right)=a^*\left( A^{(l)} \right)$.
We consider the situation that $a^*\left( A^\infty \right)$ is taken by player $j^\prime \neq j$, the action of player $j$ is $a^*\left( A^{(l^\prime)} \right)$, and actions of other players are all contained in $A^\infty$.
(We remark that such situation exists because of the definition of $A^\infty$.)
Then, there are following three cases.
\begin{enumerate}
\item $l^\prime < l$\\
For this case, player $j$ continues to play $a^*\left( A^{(l^\prime)} \right)$ in the next round, because $a^*\left( A^{(l^\prime)} \right)$ is stronger than all actions in $A^\infty \subset A^{(l^\prime)}$.

\item $l^\prime = l$\\
For this case, player $j$ continues to play $a^*\left( A^{(l)} \right)$ in the next round, because $a^*\left( A^{(l)} \right)$ is the strongest action in $A^\infty$.

\item $l^\prime > l$\\
We define the notation $-\{j, j^{\prime}\}:=\mathcal{N}\backslash \left\{ j, j^{\prime} \right\}$.
We write an action profile in the round as
\begin{align}
 \tilde{\bm{a}} &:= \left( a_j=a^*\left( A^{(l^\prime)} \right), a_{j^\prime}=a^*\left( A^\infty \right), \tilde{a}_{-\{j, j^{\prime}\}} \right)
\end{align}
with $\tilde{a}_{-\{j, j^{\prime}\}} \in A^{\infty N-2}$.
For this case, the following two situations can be considered.

First, if $s_j\left( \tilde{\bm{a}} \right)<s_{j^\prime}\left( \tilde{\bm{a}} \right)$, player $j$ switches to $a^*\left( A^\infty \right)$ with a finite probability in the next round, since $\tilde{\bm{a}}\in A^{(l)N}$ and
\begin{align}
  s_{j^\prime} \left( \tilde{\bm{a}} \right) &\geq \max_{k\neq j^\prime} s_k \left( \tilde{\bm{a}} \right).
  \label{eq:s_jp_unex}
\end{align}
(It should be noted that player $j$ may switch to $\tilde{a}_{j^{\prime\prime}}$ such that $s_j\left( \tilde{\bm{a}} \right)<s_{j^\prime}\left( \tilde{\bm{a}} \right)=s_{j^{\prime\prime}}\left( \tilde{\bm{a}} \right)$.)
However, since $a^*\left( A^\infty \right)$ is observed as an action of players $-j$ an infinite number of times, the action of player $j$ is finally absorbed to $a^*\left( A^\infty \right)$.

Second, if $s_j\left( \tilde{\bm{a}} \right)=s_{j^\prime}\left( \tilde{\bm{a}} \right)$, player $j$ continues to play $a^*\left( A^{(l^\prime)} \right)$ in the next round due to Eq. (\ref{eq:s_jp_unex}).
It should be noted that player $j$ is not beaten for the action profile $\tilde{\bm{a}}$.
If such equality holds every time player $j$ takes $a^*\left( A^{(l^\prime)} \right)$ and a player in $-j$ takes $a^*\left( A^\infty \right)$, $a^*\left( A^{(l^\prime)} \right)$ is actually equivalent to $a^*\left( A^\infty \right)$.
Otherwise, this situation is reduced to the first situation.
\end{enumerate}

Therefore, we conclude that the action of player $j$ converges to an action which is equivalent to or stronger than $a^*\left( A^\infty \right)$.
Since the payoffs in infinitely repeated games are defined by the time average of payoffs in each round (\ref{eq:payoff_repeated}), this fact implies that IIB is unbeatable.
\end{proof}

One may be interested in whether the converse of Theorem \ref{thm:unbeatable_IIB} holds or not.
For $N=2$, it has been known that the converse is true \cite{DOS2012b}.
However, for $N>2$, only the following theorem is obtained at this stage.
Here, we call the complement of unexploitable games in all multi-player totally symmetric games as \emph{exploitable games}.
In addition, we introduce the following concept.
\begin{definition}
\label{def:nodegeneracy}
In a multi-player totally symmetric game, if $a_i\neq a_j$ means $s_i(\bm{a})\neq s_j(\bm{a})$ for all pairs of players $(i,j)\in \mathcal{N}^2$ and for all $\bm{a}\in \mathcal{A}$, such game is called \emph{a game with no degeneracy}.
\end{definition}

\begin{theorem}
\label{thm:exploitable_IIB}
If the stage game is an exploitable game with no degeneracy, then IIB can be beaten.
\end{theorem}

\begin{proof}
We consider the case that player $j\in \mathcal{N}$ uses IIB.
If the stage game is an exploitable game, there exists at least one subset of the action space $A^\prime \subseteq A$ such that, for all $a_j\in A^\prime$ there exists at least one $a_{-j}\in A^{\prime N-1}$ such that
\begin{align}
 s_j \left( a_j, a_{-j} \right) - G_j\left( a_j, a_{-j} \right) &< 0.
\end{align}
For each $a_j\in A^\prime$, we write such $a_{-j}$ as $\tilde{a}_{-j}\left( a_j \right)$.
We define $k^\prime\left( a_j \right):=\arg \max_{k\neq j}s_k\left( a_j, \tilde{a}_{-j}\left( a_j \right) \right)$.
(When there are several candidates, we choose one of them.)
Because the stage game is totally symmetric, another action profile $\tilde{a}_{-j}^\prime\left( a_j \right)\in A^{\prime N-1}$ of players $-j$ in which the action of player $k^{\prime\prime}\neq j$ is exchanged for that of player $k^\prime\left( a_j \right)$ in $\tilde{a}_{-j}\left( a_j \right)$ also satisfies
\begin{align}
 s_j \left( a_j, \tilde{a}_{-j}^\prime\left( a_j \right) \right) - G_j\left( a_j, \tilde{a}_{-j}^\prime\left( a_j \right) \right) &< 0.
\end{align}
It should be noted that, for this action profile, $k^{\prime\prime}=\arg \max_{k\neq j}s_k\left( a_j, \tilde{a}_{-j}^\prime\left( a_j \right) \right)$ holds.
Then, for games with no degeneracy, when a first action of player $j$ is contained in such $A^\prime$, and players $-j$ always take $\tilde{a}_{-j}^\prime\left( a_j \right)$ for each $a_j$, the next action of player $j$ is always $a_{k^{\prime\prime}}\in A^\prime$.
That is, player $j$ always imitates the previous action of player $k^{\prime\prime}$.
Therefore, the inequality
\begin{align}
 \left\langle s_j \right\rangle^{*} &< \left\langle G_j \right\rangle^{*} = \left\langle s_{k^{\prime\prime}} \right\rangle^{*}
\end{align}
holds for such situation.
\end{proof}

Because we assumed no degeneracy, Theorem \ref{thm:exploitable_IIB} does not imply the inverse of Theorem \ref{thm:unbeatable_IIB}.
Further investigation is needed in order to clarify whether we can remove this assumption or not.

\section{Example}
\label{sec:example}
In this section, we investigate the public goods game as an example of unexploitable games.
The public goods game is one of the simplest $N$-player totally symmetric games \cite{RotKag1995,HWTN2014,PHRT2015}.
The set of actions $A$ is usually described as $A=\{ C, D \}$, where $C$ and $D$ means cooperation and defection, respectively.
The payoff of player $j\in \mathcal{N}$ is given by
\begin{align}
 s_j \left( \bm{a} \right) &= \frac{rc}{N} \sum_{k\neq j} \delta_{a_k, C} + c \left( \frac{r}{N} - 1 \right) \delta_{a_j, C} \quad (\forall \bm{a}),
\end{align}
where $c>0$ represents the unit of contribution and $1<r<N$.
It should be noted that
\begin{align}
 s_j \left( C, a_{-j} \right) - s_j \left( D, a_{-j} \right)  &= c \left( \frac{r}{N} - 1 \right) <0 \quad (\forall a_{-j}),
\end{align}
that is, $D$ dominates $C$. 

Next, for the public goods game, we calculate $G_j$ in Eq. (\ref{eq:G}).
If $D$ is contained in $a_{-j}$, $G_j$ is the payoff of the player taking $D$.
If $D$ is not contained in $a_{-j}$, $a_{-j}=(C, \cdots, C)$ and the payoffs of all players in $-j$ are equal to each other.
Therefore, we obtain
\begin{align}
 G_j \left( \bm{a} \right) &= \mathbb{I}\left( a_{-j} \neq (C, \cdots, C) \right) \frac{rc}{N} \sum_{k} \delta_{a_k, C} \nonumber \\
 & \quad + \mathbb{I}\left( a_{-j} = (C, \cdots, C) \right) \nonumber \\
 & \qquad \times \left[ \frac{rc}{N}(N-1) + \frac{rc}{N}\delta_{a_j, C} - c \right] \nonumber \\
 &= \frac{rc}{N} \sum_{k} \delta_{a_k, C} - c\mathbb{I}\left( a_{-j} = (C, \cdots, C) \right).
\end{align}
Then we find that
\begin{align}
 s_j \left( \bm{a} \right) - G_j \left( \bm{a} \right) &= -c\delta_{a_j, C} +c\mathbb{I}\left( a_{-j} = (C, \cdots, C) \right).
\end{align}
Particularly,
\begin{align}
 & s_j \left( D, a_{-j} \right) - G_j \left( D, a_{-j} \right) = c\mathbb{I}\left( a_{-j} = (C, \cdots, C) \right) \geq 0 \nonumber \\
 & \quad (\forall a_{-j}).
\end{align}
Because possible subsets of $A$ are $\{C\}$, $\{D\}$, and $\{C,D\}$, and Eq. (\ref{eq:unex}) trivially holds for $A^\prime=\{C\}, \{D\}$, this inequality means that the public goods game is an unexploitable game.
Therefore, Theorems \ref{thm:unbeatable_ZD} and \ref{thm:unbeatable_IIB} guarantee the existence of an unbeatable ZD strategy and an unbeatable IIB strategy, respectively.

According to Ref. \cite{Ued2022b}, such unbeatable ZD strategy is constructed as
\begin{align}
 \hat{T}_j \left( C | \bm{a}^\prime \right) &= -\delta_{a^\prime_j, C} + \mathbb{I}\left( a^\prime_{-j} = (C, \cdots, C) \right),
\end{align}
or
\begin{align}
 T_j \left( C | \bm{a}^\prime \right) &= \mathbb{I}\left( a^\prime_{-j} = (C, \cdots, C) \right).
\end{align}
This strategy can be regarded as a variant of TFT, which returns cooperation if and only if all other players took cooperation in the previous round.
In fact, it is reduced to TFT for the case $N=2$.

Furthermore, for the public goods game, IIB is rewritten as
\begin{align}
 T_j\left( a_j | \bm{a}^{\prime} \right) &= \delta_{a_j, a^\prime_{\arg\max_{k\neq j}s_k(\bm{a}^\prime)}} \nonumber \\
 & \qquad \times \mathbb{I} \left( -\delta_{a^\prime_j, C} +\mathbb{I}\left( a^\prime_{-j} = (C, \cdots, C) \right) <0 \right) \nonumber \\
 & \quad + \delta_{a_j, a^\prime_j} \nonumber \\
 & \qquad \times \mathbb{I} \left( -\delta_{a^\prime_j, C} +\mathbb{I}\left( a^\prime_{-j} = (C, \cdots, C) \right) \geq 0 \right) \nonumber \\
 &= \delta_{a_j, D} \mathbb{I} \left( \mathbb{I}\left( a^\prime_{-j} = (C, \cdots, C) \right) =0 \right) \nonumber \\
 & \qquad \times \mathbb{I} \left( \delta_{a^\prime_j, C} =1 \right) \nonumber \\
 & \quad + \delta_{a_j, a^\prime_j} \left[ 1-\mathbb{I} \left( \mathbb{I}\left( a^\prime_{-j} = (C, \cdots, C) \right) =0 \right) \right. \nonumber \\
 & \qquad \left. \times \mathbb{I} \left( \delta_{a^\prime_j, C} =1 \right) \right] \nonumber \\
 &= \delta_{a_j, D} \mathbb{I} \left( a^\prime_{-j} \neq (C, \cdots, C) \right) \delta_{a^\prime_j, C} + \delta_{a_j, a^\prime_j} \nonumber \\
 & \quad - \delta_{a_j, a^\prime_j} \mathbb{I} \left( a^\prime_{-j} \neq (C, \cdots, C) \right) \delta_{a^\prime_j, C}.
\end{align}
Particularly, 
\begin{align}
 T_j\left( C | \bm{a}^{\prime} \right) &=  \delta_{C, a^\prime_j} - \delta_{C, a^\prime_j} \mathbb{I} \left( a^\prime_{-j} \neq (C, \cdots, C) \right) \nonumber \\
 &= \delta_{C, a^\prime_j} \mathbb{I} \left( a^\prime_{-j} = (C, \cdots, C) \right).
\end{align}
This is nothing but the Trigger strategy \cite{Fri1971}, which returns $C$ as long as all players take $C$ in the previous round and forms a cooperative Nash equilibrium.

\section{Concluding remarks}
\label{sec:conclusion}
In this paper, we introduced the concept of unexploitable games as a class of $N$-player totally symmetric games with $N\geq 2$, which is an extension of non-generalized-rock-paper-scissors games for $N=2$.
We then proved that, for infinitely repeated non-trivial unexploitable games, there exist unbeatable ZD strategies.
In addition, we also proved that, for infinitely repeated unexploitable games, the IIB strategy is unbeatable.
These results are a natural extension of ones for $N=2$.
We also showed that the public goods game is an unexploitable game, and constructed an unbeatable ZD strategy.
For the public goods game, unbeatable IIB is also equivalent to the Trigger strategy.

Although we investigated only repeated games in which the same stage game is infinitely repeated, there are many situations where the stage game changes over time or the stage game depends on past plays, such as board games.
Investigating successfulness of imitation in these dynamic games is a challenging future problem.

Before ending this paper, we make three remarks.
The first remark is related to necessary conditions for IIB to be unbeatable.
Although unexploitable property of the stage game is a sufficient condition for IIB to be unbeatable, we could not prove that it is also a necessary condition at this stage.
However, for the case $N=2$, it is a necessary and sufficient condition \cite{DOS2012b}.
Further investigation is needed to specify a necessary and sufficient condition for IIB to be unbeatable.
In addition, one may expect that the existence of unbeatable ZD strategies and that of unbeatable imitation strategies are related to each other.
Techniques used for the proofs of the existence in this paper seem to be similar but different, as we used the existence of the strongest action and the weakest action for the proof of the former, and the order of strength of actions for the proof of the latter.
Clarifying the relation between these two strategy classes is the subject of future work.

The second remark is one about finiteness of the action space.
In this paper, we assumed that the set of actions $A$ is finite.
However, there are many situations that $A$ is infinite, such as the Cournot oligopoly game \cite{FudTir1991}.
Since we have used that $A$ is finite in the proofs of Theorems \ref{thm:unbeatable_ZD} and \ref{thm:unbeatable_IIB}, we do not know these theorems can be straightforwardly extended to the case that $A$ is infinite.
We would like to investigate whether these results can be extended to the case that $A$ is infinite or not, in future.

The third remark is the relation between unexploitable games and potential games \cite{MonSha1996}.
For the case $N=2$, potential games are a special case of unexploitable games \cite{DOS2012b}.
However, for its proof, a special property of $N=2$ was used.
For $N>2$, the relation between unexploitable games and potential games is not clear.
In fact, it has been known that, for the Cournot oligopoly game, which is an example of totally symmetric potential games with infinite action space, IIB can be beaten \cite{DOS2012b}.
We are interested in clarifying the relation between unexploitable games and potential games for general $N\geq 2$.


\bibliographystyle{IEEEtran}
\bibliography{unexploitable}

\EOD

\end{document}